\documentclass{article}
\usepackage{graphicx,amssymb,amsmath}
\usepackage{url,float}
\usepackage{amsfonts}
\usepackage{latexsym}
\usepackage[boxruled]{algorithm2e}
\usepackage{tcolorbox} %textbox

\usepackage{color}

\newcommand{\remove}[1]{{}}

\usepackage[normalem]{ulem}

\newcommand{\ABox}{
\raisebox{3pt}{\framebox[6pt]{\rule{6pt}{0pt}}}
}
\newenvironment{proof}{{\bf Proof:}}{\hfill\ABox}

\newtheorem{theorem}{{\bf Theorem}}

\newtheorem{lemma}{Lemma}

\newcommand{\lemlab}[1]{\label{lemma:#1}}
\newcommand{\thmlab}[1]{\label{thm:#1}}

\newcommand{\figlab}[1]{\label{fig:#1}}
\newcommand{\seclab}[1]{\label{sec:#1}}

\newcommand{\lemref}[1]{\ref{lemma:#1}}

\newcommand{\figref}[1]{\ref{fig:#1}}
\def\P{{\mathcal P}}

\def\C{{\mathcal C}}
\def\bcC{{\partial \mathcal C}}
\def\bC{{\partial C}}

\def\cF{{\mathcal F}}
\def\g{{\gamma}}

\def\F{{\Phi}}

\def\o{{\omega}}
\def\O{{\Omega}}
\def\D{{\Delta}}
\def\d{{\delta}}
\def\e{{\varepsilon}}

\def\q{{\theta}}

\newcommand{\squeezelist}{\setlength{\itemsep}{0pt}}
\title{Addendum to:\\
Edge-Unfolding Nearly Flat Convex Caps}

\author{
Joseph O'Rourke%
    \thanks{Department of Computer Science, Smith College, Northampton, MA, USA.
      \protect\url{jorourke@smith.edu}.}
}

\index{O'Rourke, Joseph}

\begin{document}
\maketitle

\begin{abstract}
This addendum to~\cite{o-eunfcc-17} establishes that 
a nearly flat acutely triangulated convex cap in the sense of that paper can be edge-unfolded
even if closed to a polyhedron by adding the convex polygonal base under the cap.
\end{abstract}

\section{Introduction}
\seclab{Introduction}
The paper~\cite{o-eunfcc-17} established that every sufficiently flat acutely triangulated
convex cap has an edge-unfolding to a non-overlapping simple polygon, i.e., a \emph{net}.
I used the term ``convex cap'' in the following sense
(where $\phi(f)$ is
the angle the normal to face $f$ makes with the $z$-axis):
\begin{quotation}
\noindent
Define a \emph{convex cap} $\C$ of angle $\Phi$ to be $C=\P \cap H$
for some $\P$ and $H$, such that $\phi(f) \le \Phi$ for all $f$ in $\C$.
%We will only consider $\Phi < 90^\circ$, which implies that the projection
%$C$ of $\C$ onto the $xy$-plane is one-to-one.
[...]
Note that $\C$ is not a closed polyhedron; it has no ``bottom,''
but rather a boundary $\bcC$.
\end{quotation}
This note proves that same claim holds even when $\C$ is closed to a polyhedron
by adjoining the convex base face $B$ bounded by $\bcC$.
Eventually this addendum will be incorporated into a future version~\cite{o-eunfcc-17}.
For now we assume familiarity with that paper, and especially the 
section below, the most relevant portions of which we reproduce verbatim. Ellisons are 
marked by ``[...].''

\section{Angle-Monotone Spanning Forest}
\seclab{AngMonoForest}
%\begin{tcolorbox}
%2.~Generalizing the result in~\cite{lo-ampnt-17},
%there is a $\q$-angle-monotone, boundary rooted spanning forest $F$ of $C$, for $\q < 90^\circ$.
%$F$ lifts to a spanning forest $\cF$ of the convex cap $\C$.
%\end{tcolorbox}
%
%First we define angle-monotone paths, 
%which originated in~\cite{dfg-icgps-15}
%and were further developed in~\cite{bbcklv-gtamg-16},
%and then turn to the spanning forests we need here.
%
%\subsection{Angle-Monotone Paths}
%\seclab{AngMonoPaths}
%Let $C$ be a planar, triangulated convex domain,
%with $\bC$ its boundary, a convex polygon. 
%Let $G$ be the (geometric) graph of all the triangulation edges in $C$ and on $\bC$.
%
%Define the \emph{$\q$-wedge} $W(\b,v)$ to be the region of
%the plane 
%bounded by rays at angles $\b$ and $\b + \q$ 
%emanating from $v$.
%$W$ is closed along (i.e., includes) both rays, and has 
%angular \emph{width} of $\q$.
%A polygonal path $Q=(v_0,\ldots,v_k)$ following edges of $G$ is called \emph{$\q$-angle-monotone}
%(or \emph{$\q$-monotone} for short)
%if the vector of every edge $(v_i, v_{i+1})$ lies in $W(\b,v_0)$
%(and therefore $Q \subseteq W(\b,v_0)$), for some $\b$.\footnote{
%My notation here is slightly different from the notation in~\cite{lo-ampnt-17} and earlier papers.
%}
%Note that if $\b=0^\circ$ and $\q=90^\circ$, then a $\q$-monotone path
%is both $x$- and $y$-monotone, i.e., it meets every vertical, and every horizontal line
%in a point or a segment, or not at all.

\noindent
[...]

\subsection{Angle-Monotone Spanning Forest}
\seclab{SpanningForest}
``It was proved
in~\cite{lo-ampnt-17}
that every nonobtuse triangulation $G$ of a convex region $C$
has a boundary-rooted spanning forest $F$ of $C$, with all paths in $F$
$90^\circ$-monotone.
We describe the proof and simple construction algorithm before detailing
the changes necessary for acute triangulations.

Some internal vertex $q$ of $G$ is selected, and the plane partitioned into
four $90^\circ$-quadrants $Q_0,Q_1,Q_2,Q_3$ by orthogonal lines through $v$.
Each quadrant is closed along one axis and open on its counterclockwise axis;
$q$ is considered in $Q_0$ and not in the others, so the quadrants partition the plane.
It will simplify matters later if we orient the axes so that no vertex except
for $q$ lies on the axes, which is clearly always possible.
Then paths are grown within each quadrant independently, as follows.
A path is grown from any vertex $v \in Q_i$ not yet included in the forest $F_i$,
stopping when it reaches either a vertex already in $F_i$, or $\bC$.
These paths never leave $Q_i$, and result in a forest $F_i$ spanning the vertices in $Q_i$ .
No cycle can occur because a path is grown from $v$ only when $v$ is not already
in $F_i$; so $v$ becomes a leaf of a tree in $F_i$.
Then $F = F_1 \cup F_2 \cup F_3 \cup F_4$.

We cannot follow this construction exactly in our situation of an acute triangulation $G$,
because the ``quadrants" for $\q$-monotone paths for $\q = 90^\circ - \D\q < 90^\circ$ cannot cover the plane
exactly: They leave a thin $4 \D\q$ angular gap; call the cone of this aperature $g$.
We proceed as follows.
Identify an internal vertex $q$ of $G$ so that it is possible to
orient the cone-gap $g$, apexed at $q$, so that $g$ contains no internal vertices of $G$.
See Fig.~\figref{QuadGap} for an example. Then we proceed just
as in~\cite{lo-ampnt-17}: paths are grown within each $Q_i$, forming four
forests $F_i$, each composed of $\q$-monotone paths.
%%%%%%%%%%%%%%%%%%%%%%%%%%%%%%%%%Figure Begin
\begin{figure}[htbp]
\centering
\includegraphics[width=0.6\linewidth]{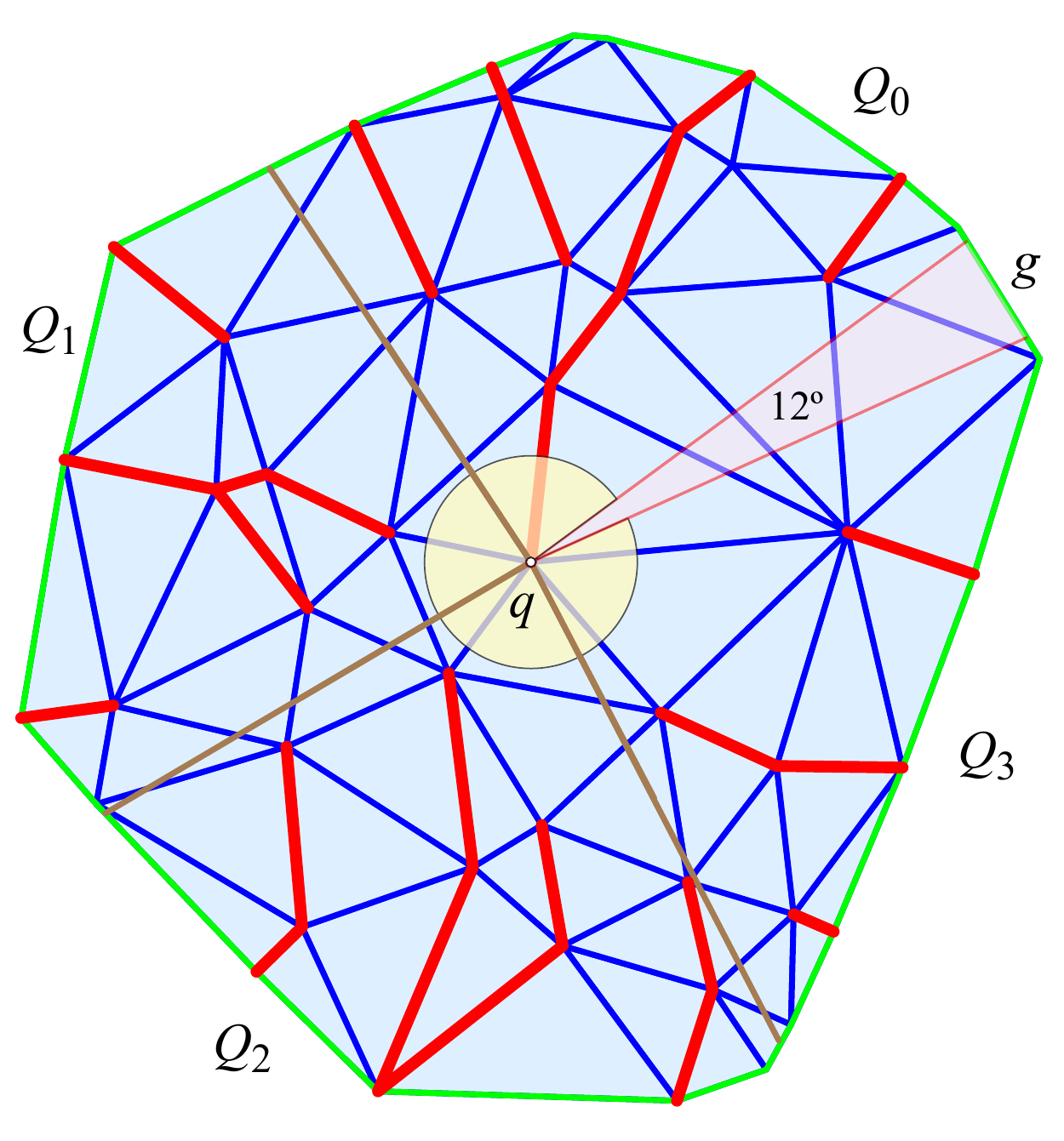}
%\fbox{Figures/xxx.pdf}
\caption{Here the near-quadrants $Q_i$ have width $\q=87^\circ$, 
so the gap $g$ has angle $4\D\q  = 12^\circ$.}
\figlab{QuadGap}
\end{figure}
%%%%%%%%%%%%%%%%%%%%%%%%%%%%%%%%%Figure End

It remains to argue that there always is such a $q$ at which to apex cone-gap $g$.
Although it is natural to imagine $q$ as centrally located (as in Fig.~\figref{QuadGap}),
it is possible that $G$ is so dense with vertices that such a central location is not possible.
However, it is clear that the vertex $q$ that is closest to $\bC$ will suffice: aim $g$
along the shortest path from $q$ to $\bC$. Then $g$ might include several vertices on $\bC$,
but it cannot contain any internal vertices of $G$, as they would be closer to $\bC$.
Again we could rotate the axes slightly so that no vertex except for $q$ lies on an axis.''

%We conclude this section with a lemma:
%\begin{lemma}
%If $G$ is an acute triangulation of a convex region $C$,
%then there exists a 
%boundary-rooted spanning forest $F$ of $C$, with all paths in $F$
%$\q$-angle-monotone, for $\q = 90^\circ - \D\q < 90^\circ$.
%\lemlab{QuadGap}
%\end{lemma}
%%\begin{proof}
%%\end{proof}
%
%\noindent
%That $F$ lifts to a spanning forest $\cF$ of the convex cap $\C$ is immediate.
%What is not straightforward is establishing the requisite properties of $\cF$.

\noindent
[...] (End quoted text.)

\section{Unfolding the Base $B$}
\seclab{Base}
By our definition of a convex cap, its boundary $\bcC$ lies in a plane, and so bounds
a convex polygonal base $B$.
We assume that, unlike the cap $\C$, $B$ is not triangulated, and so must be unfolded
as an intact unit.
(Of course, it can be unfolded as a unit even if triangulated.)

Let $C_\perp$ be the unfolded net of $\C$ produced by the algorithm in~\cite{o-eunfcc-17}.
If some edge $e$ of $C_\perp$ lies on the convex hull of $C_\perp$, then $B$ can be
``flipped out'' to $B'$ around $e$ by cutting all edges of $\bcC$ except for $e$.
Because $B'$ is convex and is attached to the hull of $C_\perp$, it is clear there is
no overlap, and we would be finished. In fact this is the proof path we will follow, but it is
not as straightforward as it might seem.

\subsection{Obstructions}
\seclab{Obstructions}
We now argue that there is an arbitrarily flat convex cap $\C$ and a spanning
cut forest $\cF$ such that there is no such edge $e$ of  $C_\perp$
to which to attach $B'$ without overlap.
First, we look at a ``real'' unfolding to see what form the obstruction might take.
Fig.~\figref{AM_20_30_s2_v83_ell_Hull} shows a portion of $C_\perp$, 
identifying a particular edge $e$ which is tilted inside the hull and would lead
to overlap were $B'$ attached there.
%see Fig.~\figref{AM_20_30_s2_v83_ell_Hull}.
%%%%%%%%%%%%%%%%%%%%%%%%%%%%%%%%%Figure Begin
\begin{figure}[htbp]
\centering
\includegraphics[width=0.75\linewidth]{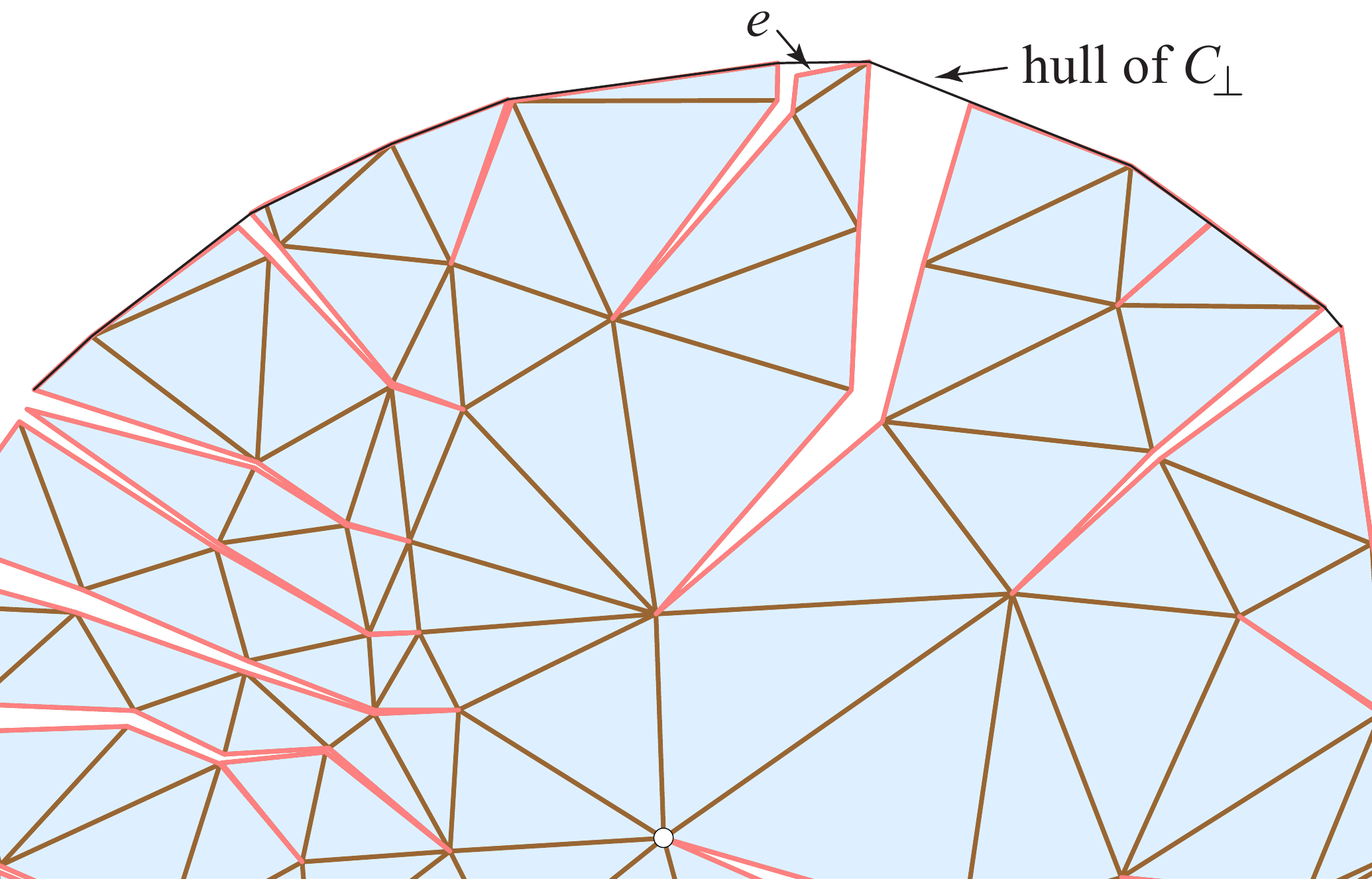}
%\fbox{Figures/xxx.pdf}
\caption{Detail from Fig.~24 of~\protect\cite{o-eunfcc-17}, with
a portion of the convex hull marked.}
\figlab{AM_20_30_s2_v83_ell_Hull}
\end{figure}
%%%%%%%%%%%%%%%%%%%%%%%%%%%%%%%%%Figure End
However, even in this example, there are many other candidates for $e$ that would suffice as
$B$'s attachment. This suggests the next question: Is there a cap $\C$ 
and a cut forest $\cF$ such that
every edge of $C_\perp$ is similarly titled inside the hull, leaving no ``safe'' attachment
edge for $B$?
The answer is {\sc yes}. We only sketch the argument before discussing in
more detail how to circumvent this counterexample.

Let $\C$ be a cap whose boundary $\bcC$ is a $12$-sided regular polygon (i.e.,
a dodecagon). For the construction to work, we need at least a $9$-sided polygon; $12$
makes it visually clearer.
The cut forest $\cF$ is as illustrated in Fig.~\figref{BaseCex_22}:
one tree in $\cF$ is a $2$-path from the center of $\C$,
and all the others trees are single segments. The key property is that each
segment creates a very shallow angle with $\bcC$.
%see Fig.~\figref{BaseCex_22}.
%%%%%%%%%%%%%%%%%%%%%%%%%%%%%%%%%Figure Begin
\begin{figure}[htbp]
\centering
\includegraphics[width=0.7\linewidth]{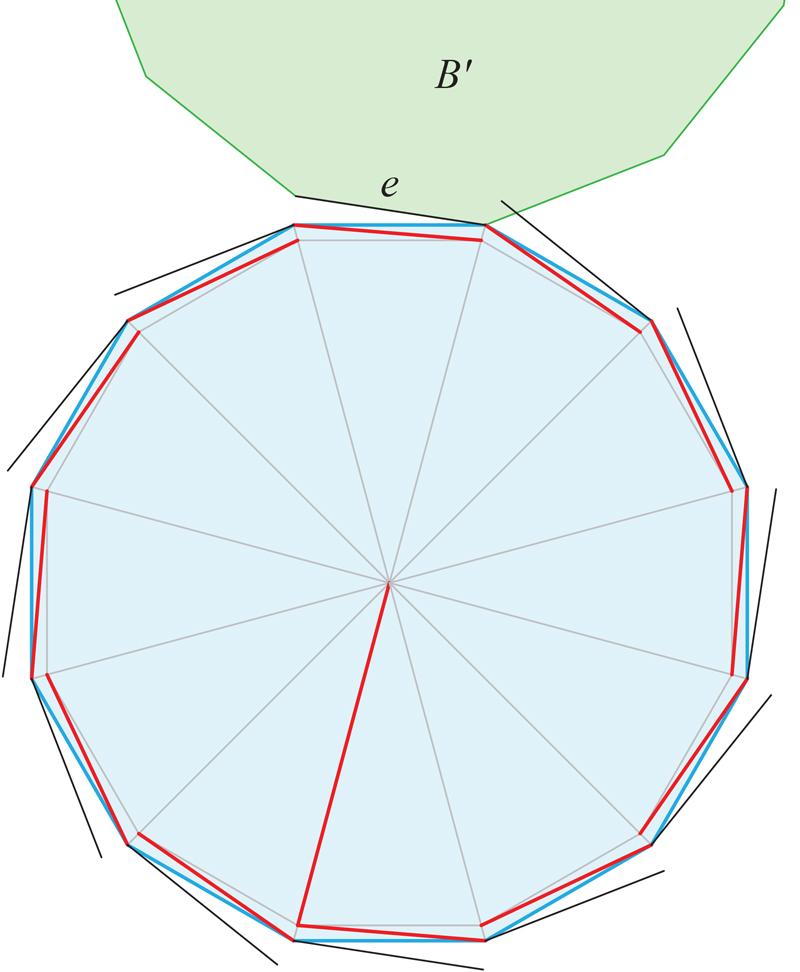}
%\fbox{Figures/xxx.pdf}
\caption{No edge of $C_\perp$ is a convex hull edge. The cut forest $F$ is shown red,
$\bcC$ is blue, developed edges of $C_\perp$ black.}
\figlab{BaseCex_22}
\end{figure}
%%%%%%%%%%%%%%%%%%%%%%%%%%%%%%%%%Figure End

We arrange near-zero curvature at the central vertex, and all the other vertices have
the same curvature $\o > 0$. Because of the shallow angle they form with $\bcC$,
the opening gap caused by cutting each segment is nearly orthogonal to the cut segment.
With the internal angle of a $12$-gon $\frac{5}{6} \pi$, the exterior angle between
$B$ and a reflected copy $B'$ of $B$ is $\frac{2}{6} \pi = 60^\circ < 90^\circ$.
This allows the orthogonally jutting rotation of each boundary edge to penetrate into a
reflected $B'$, reflected about the next edge $e$ counterclockwise,
as illustrated.

There is no impediment to realizing this example in 3D so that the curvatures $\o$ suffice
to render every boundary edge of $C_\perp$ leading to overlap with $B'$. Moreover, this could
be accomplished for an arbitrarily flat $\C$ by increasing the number of sides of the $n$-gon base
so that even a very small $\o$ results in overlap.
These claims will not be justified further, as they only serve to motivate the next steps.

\subsection{Quadrant-based forest $F$}
\seclab{Quadrant-based}
The reason the preceding counterexample does not present an insurmountable
obstacle is that the spanning forest $F$ selected in the planar projection graph $G$ of $\C$
is not arbitrary, but instead is based on the quadrants
illustrated in Fig.~\figref{QuadGap}.
We now show that the quadrant-based forest $F$ leads to an edge $e$ of $C_\perp$
to which to attach the reflected base $B'$.

A reminder on the notation we are employing: $\C$ is the convex cap in $\mathbb{R}^3$,
$C$ is its projection on the $xy$-plane, $F$ the spanning forest in that plane,
$\cF$ the lift of the forest, and $C_\perp$ is the development of the cap $\C$ 
after cutting $\cF$.

Let $v \in \bC$ be a vertex on the boundary of the projection $C$ that is
a root of a tree in the forest $F$. Rather than
viewing the lift, cut, and development as producing $C_\perp$,
it will help to view the movement of $v$ in the plane caused by opening
the curvatures along the cut paths terminating at $v$.
As we saw in Section~8 %\secref{LandR} 
of~\cite{o-eunfcc-17},
we can view each cut path $Q$ to $v$ as two planar polygonal chains $L$ and $R$
which are initially identical, and then open at each vertex $v_i$ along
$Q$ by the curvatures $\o_i$. Here we are only interested in the final planar displacement of
the root $v$, the endpoint of the chain. Let $v$ and $v'$ be the original and displaced
versions of $v$, i.e., the last vertices of $L$ and $R$. The \emph{gap segment} $v v'$ represents
the gap at the boundary of $C_\perp$ caused by opening the cuts in $F$ to $v$,
visible, for example, in Fig.~\figref{AM_20_30_s2_v83_ell_Hull}.

The gap segment $v v'$ is caused by the composition of several (small) rotations about
different centers, the vertices along $Q$. It is well-known that $v v'$ is equivalent
to a single rotation about a (generally) different center $c$,
which we'll call the \emph{composite center of rotation}.
We claim that, for sufficiently small $\o_i$, $c$ is either inside the convex hull of $Q$,
or arbitrarily close to the boundary of the hull.
This claim is justified in the Appendix by
Lemma~\lemref{center-of-gravity}.

Returning to Fig.~\figref{BaseCex_22}, the centers of rotation in $F$ were
arranged so that the gap segments were nearly orthogonal to $\bC$,
so that they ``jutted out'' and caused overlap with the reflected $B'$.
We now argue first, that a different arrangement of centers of rotation
can produce ``safe'' gap segments, and second, that this can be achieved by
a quadrant-based spanning forest $F$. 

Arrange an edge $e=(v,u)$ of $\bC$ to be topmost and horizontal
and crossing the vertical quadrant boundary, and suppose
both $v$ and $u$ are roots of cut trees in $F$.
We call $e$ a locally \emph{safe edge} if the composite centers of rotation $c_v$ and $c_u$
for the trees incident to $v$ and $u$ fall underneath $e$, as illustrated in  
Fig.~\figref{SafeEdge}. For then the gap segments angle down below $e$,
making $e$ a safe candidate for the attachment of $B'$.
%%%%%%%%%%%%%%%%%%%%%%%%%%%%%%%%%Figure Begin
\begin{figure}[htbp]
\centering
\includegraphics[width=0.9\linewidth]{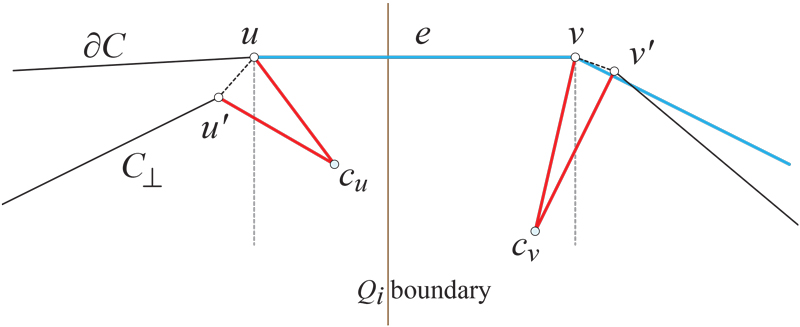}
%\fbox{Figures/xxx.pdf}
\caption{Safe edge $e$.}
\figlab{SafeEdge}
\end{figure}
%%%%%%%%%%%%%%%%%%%%%%%%%%%%%%%%%Figure End

Returning to Fig.~\figref{QuadGap}, we now argue that what was there called the gap edge $g$ 
of $\bC$ can serve as a safe edge $e$.
As in that figure, let $q$ be the quadrants origin. The shortest path $\g$ on $\C$ to $\bcC$
is orthogonal to a boundary edge $e$. The development of the geodesic $\g$ is a straight
line. Use this straight line as the vertical axis of the quadrants, with
$e$ horizontal.
Now, because the spanning forest algorithm grows edges within the quadrant wedges,
we are guaranteed that all edges of a tree incident to the endpoints of $e$
are slanted such that they angle strictly vertically underneath $e$, as
illustrated in Fig.~\figref{QuadTrees}. The strictness follows because the wedges
are $\D\q$ less than $90^\circ$.
%see Fig.~\figref{QuadTrees}.
%%%%%%%%%%%%%%%%%%%%%%%%%%%%%%%%%Figure Begin
\begin{figure}[htbp]
\centering
\includegraphics[width=0.75\linewidth]{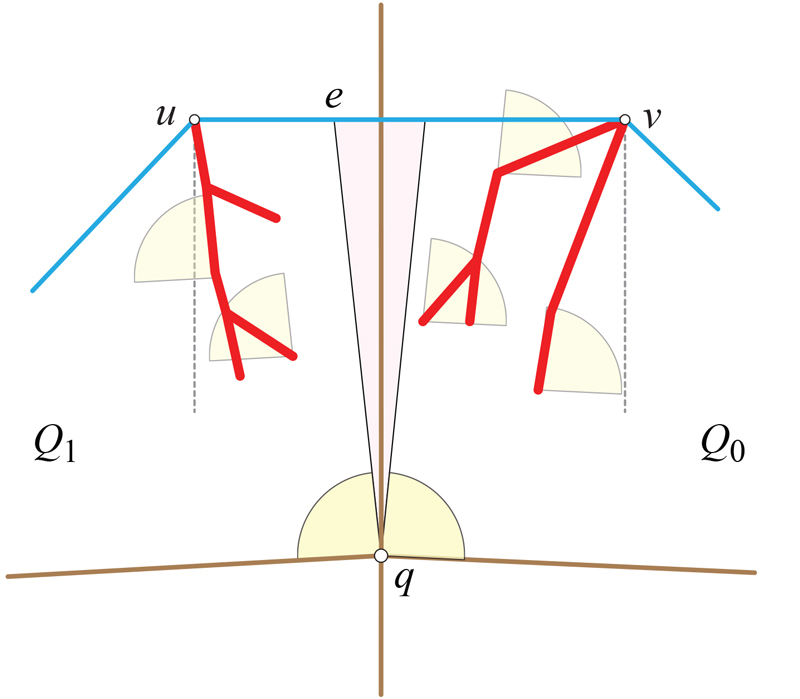}
%\fbox{Figures/xxx.pdf}
\caption{The trees incident to the endpoints of $e$ have composite centers of rotation
underneath $e$.}
\figlab{QuadTrees}
\end{figure}
%%%%%%%%%%%%%%%%%%%%%%%%%%%%%%%%%Figure End

Now applying Lemma~\lemref{center-of-gravity}, we obtain that the
composite center of rotation is underneath $e$.
Even though that lemma allows the true center to be slightly outside the
hull of the rotation centers, that the wedges have angle less than $90^\circ$
permits the conclusion that the true center is in the hull for sufficiently small
$\o_i$, and therefore strictly underneath $e$.
Thus we know that $e$ is a locally safe edge.

But now it is easy to reapply this argument to conclude that none of the gap segments
for other vertices around $\bC$ are angled above the horizontal, and so $e$ is in
fact globally safe. Thus we can attach the reflected base $B'$ along $e$ without overlap.
This proves:
\begin{theorem}
A convex polyhedron consisting of a
a nearly flat acutely triangulated convex cap $\C$ joined along $\bcC$ to a
base $B$ can be edge-unfolded without overlap,
for sufficiently small cap curvature $\O$.
\thmlab{BaseSafe}
\end{theorem}

%\noindent
%Unlike Theorem~1 %\thmref{CCapEUnf}
%of~\cite{o-eunfcc-17}, I have not yet derived an explicit bound
%on the $\F$ (and therefore $\O$) that suffices for a given acuteness gap $\D\q$,
%which is why the statement can only be claimed for ``sufficiently small curvature.''

\noindent
Using the error between the true and approximate composite rotation centers
$\d=\frac{1}{2} \sum_i \ell_i \o_i$ from Lemma~\lemref{RotationsCG},
and crudely summarizing this as $\d=L \o / 2$ for a total chain length $L$,
a calculation shows the wedge slant $\D\q$ leads to 
``sufficiently small'' curvatures 
satisfied if $\o \lesssim 2 \D\q$.
But already we know that 
$$
\O < \pi \F^2 < \pi (0.3 \sqrt{\D\q} )^2 \approx 0.28 \D\q \;,
$$
and because $\o < \O$ for any one tree of $F$, the curvatures are already
``sufficiently small'' from other constraints.

An illustration is shown in Fig.~\figref{CB_10_20_s3_3D_Unf}.
The selection of $e$ in this example does not follow the proof exactly
just due to limitations of my implementation
($q$ is not closest to $\bcC$ and $e$ is not orthogonal to the vertical quadrant axis), but it illustrates how $e$ is locally
and indeed globally safe.
(In this and in most examples, there are many safe edges.)
%see Fig.~\figref{CB_10_20_s3_3D_Unf}.
%%%%%%%%%%%%%%%%%%%%%%%%%%%%%%%%%Figure Begin
\begin{figure}[htbp]
\centering
\includegraphics[width=1.0\linewidth]{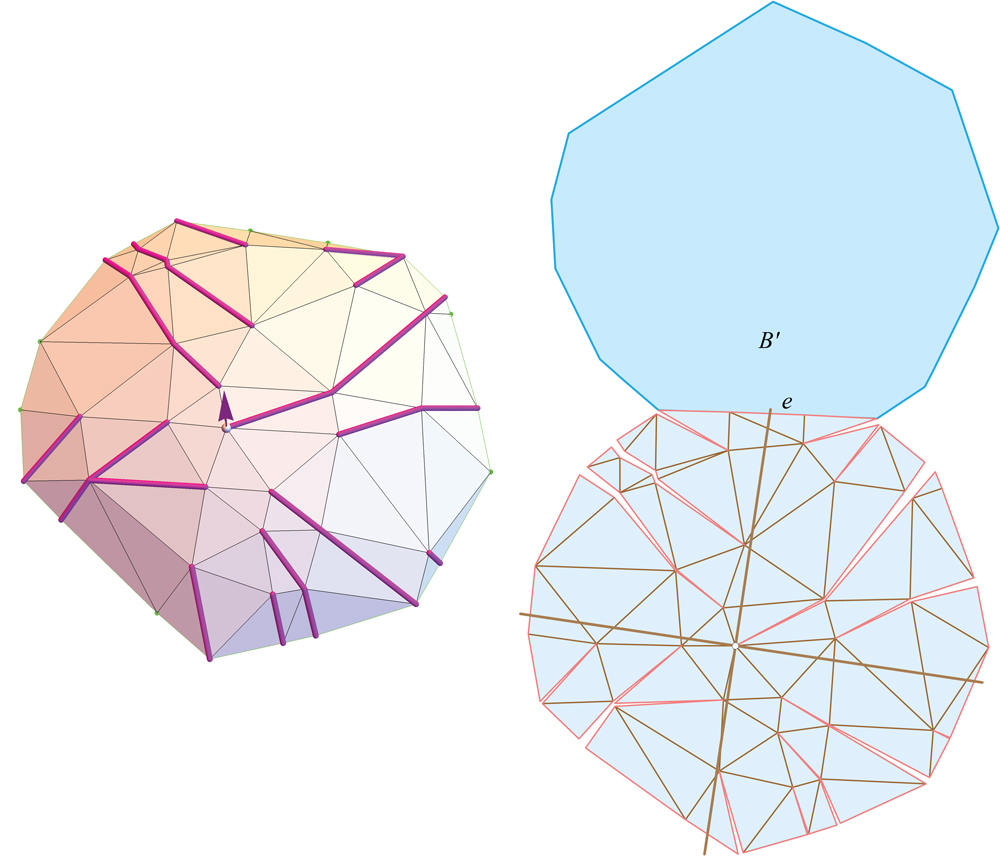}
%\fbox{Figures/xxx.pdf}
\caption{Cap $\C$ (left) and an edge-unfolding (right), 
including base $B$ flipped across safe edge $e$.}
\figlab{CB_10_20_s3_3D_Unf}
\end{figure}
%%%%%%%%%%%%%%%%%%%%%%%%%%%%%%%%%Figure End

%%%%%%%%%APPENDIX%%%%%%%%
\newpage
\section*{Appendix}
We need a lemma that allows us to conclude that, for small curvatures,
the effect of the rotations along a cut path $Q$ to a boundary vertex $v \in \bC$
is equivalent to one rotation from a point in the convex hull of the vertices along $Q$.

\begin{lemma}
Let $R_i(\o_i,p_i)$ be a two-dimensional rotation by angle $\o_i \ge 0$
about point $p_i$, for $i=1,\ldots,k$.
Then, for sufficiently small $\o_i$,
the result of composing the $k$ rotations $R_i$ is 
equivalent to one rotation about a 
\emph{center-of-gravity} rotation center: 
the sum of the $p_i$
weighted by the angles:
$$ 
R_1(\e \o_1,p_1) \circ \cdots \circ R_k(\e \o_k,p_k)
\rightarrow R(\e \o,p)
$$
as $\e \to 0$, where
$$\o = \sum_i \o_i $$
and
$$p = (\o_1 p_1 + \cdots + \o_k p_k) /  \o \;.$$ %\sum_i \o_i \;. $$
\lemlab{RotationsCG}
\end{lemma}

\noindent
The role of $\e$ is to ensure all the angles approach $0$.
Equivalently (and more appropriate in our context), we can just think of
the $\o_i$ as ``sufficiently small.''
This proposition is illustrated in Fig.~\figref{RotationCG}
for a polygonal chain.
%see Fig.~\figref{RotationCG}.
%%%%%%%%%%%%%%%%%%%%%%%%%%%%%%%%%Figure Begin
\begin{figure}[htbp]
\centering
\includegraphics[width=1.0\linewidth]{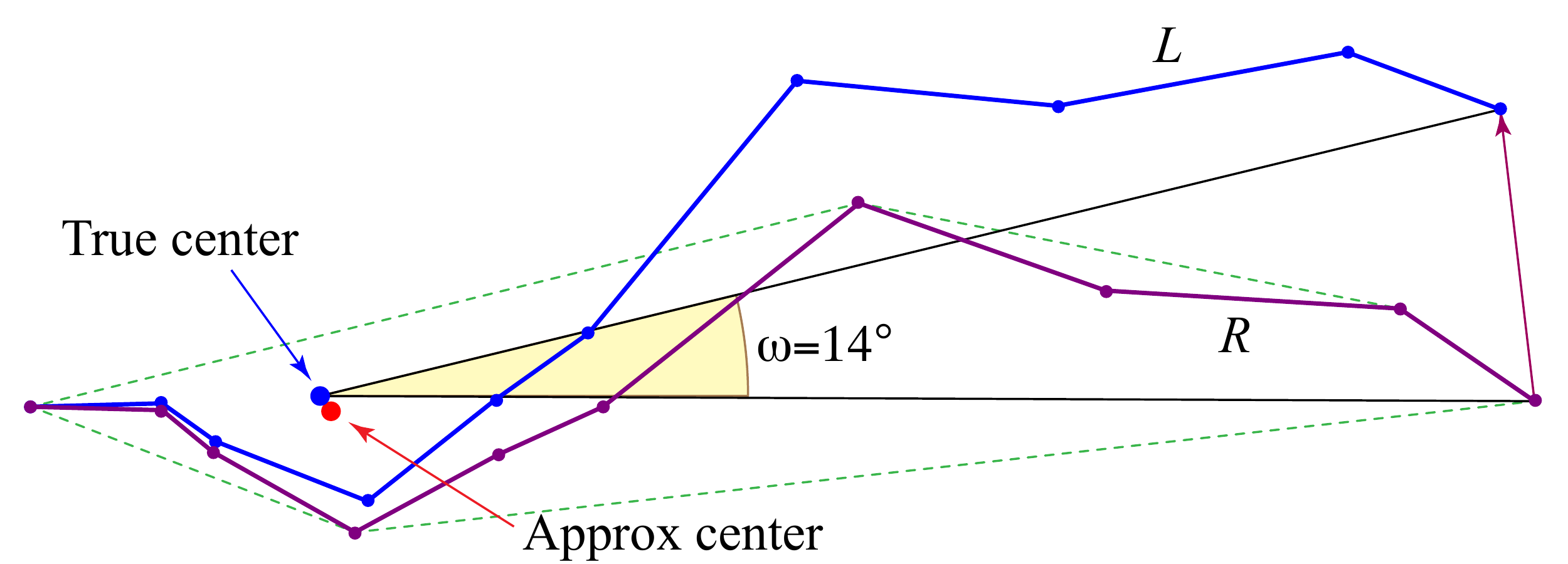}
%\fbox{Figures/xxx.pdf}
\caption{Comparison of true composite center of rotation, and the 
approximate center-of-gravity center. Here $R$ is fixed and $L$ obtained
by $\o_i$ rotations.}
\figlab{RotationCG}
\end{figure}
%%%%%%%%%%%%%%%%%%%%%%%%%%%%%%%%%Figure End

\noindent
\begin{proof}
%\emph{(Sketch.)}
It is well-known that the composition of two rotations by angles $\o_1, \o_2$
about different centers $p_1, p_2$ is equivalent to
one rotation by $\o_1 + \o_2$ about a (generally) different center $c$.\footnote{
Unless $\o_1 + \o_2 = 2 \pi$, which will never occur with small rotations.}
Consequently, the same holds for the composition of $k$ rotations.
We now prove that as $\o_1, \o_2$ approach $0$, the center $c$ approaches the point
$p=(\o_1 p_1 + \o_2 p_2) / (\o_1 + \o_2)$ on the $p_1 p_2$ segment.
Following~\cite[p.38]{n-vca-98}, we view the rotations by $\o_1, \o_2$
as reflections in lines separated by $\o_1 /2, \o_2 /2$.
Then $c$ is the intersection of two reflection lines, as illustrated in Fig.~\figref{CompRotErr}.
With $p_1=(0,0)$ and $p_2=(1,0)$, explicit calculation yields
$$c = \left( \,
\frac{\sin \o_2}{\sin \o_1 + \sin \o_2} \, , \, 
\frac{\sin \o_1 \sin \o_2}{\sin \o_1 + \sin \o_2}  \, 
\right) \;.
$$
From this expression and that for $p$ above, futher calculation shows
that the error $\d = | c - p |$ is $\frac{1}{8} (\o_1 +  \o_2)$ for small $\o_i$.
So indeed $\d$ approaches zero.

Repeating the argument for $k$ rotations yields 
(via a calculation not shown here)
that the error $\d$ is bounded by $\frac{1}{2} \sum_i \ell_i \o_i$,
where $\ell_i=| p_{i+1}-p_i |$ are the link lengths of the chain,
as $\o_i \to 0$.
Thus, $\d \to 0$, $c$ approaches $p$, and
the claim of the lemma is established.
\end{proof}

%\cite{n-vca-98}
%see Fig.~\figref{CompRotErr}.
%%%%%%%%%%%%%%%%%%%%%%%%%%%%%%%%%Figure Begin
\begin{figure}[htbp]
\centering
\includegraphics[width=0.75\linewidth]{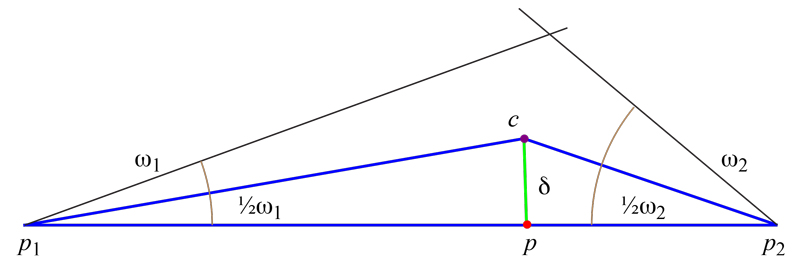}
%\fbox{Figures/xxx.pdf}
\caption{The error $\d$ between the true composite center $c$ and the center-of-gravity center.}
\figlab{CompRotErr}
\end{figure}
%%%%%%%%%%%%%%%%%%%%%%%%%%%%%%%%%Figure End

\noindent
An immediate implication of Lemma~\lemref{RotationsCG} is:
\begin{lemma}
Under the same assumptions, the center-of-gravity
approximate center $p$ 
approaches a point in the convex hull of $\{ p_1, \ldots, p_k \}$
as $\e \to 0$, or equivalently, as $\o \to 0$.
\lemlab{center-of-gravity}
\end{lemma}
\begin{proof}
With $\o_i \ge 0$, the weighted sum in Lemma~\lemref{RotationsCG}
is a convex combination of the $p_i$ points,
and so inside (or on the boundary of) the convex hull.
\end{proof}

%\noindent
%I have not computed explicit bounds on the difference between the
%true center and the %center-of-gravity
%approximate center, but we know this difference approaches $0$ as $\o \to 0$.

\newpage 
\bibliographystyle{alpha}
\bibliography{CCapEUnf}
\end{document}